\documentclass[11pt,letterpaper]{article}

\usepackage[margin=1in]{geometry}
\usepackage{graphicx}
\usepackage{booktabs}
\usepackage{amsthm}
\usepackage{amsmath,amsfonts,amssymb}
\usepackage{thmtools}
\usepackage{thm-restate}
\usepackage{doi}
\usepackage{hyperref}
\usepackage[capitalise]{cleveref}
\usepackage{comment}
\usepackage{multirow}
\newcommand{\citep}{\cite}
\newcommand{\citet}{\cite}

\newcommand{\RR}{\mathbb{R}}
\newcommand{\Rd}{\RR^d}
\newcommand{\norm}[1]{\left\lVert#1\right\rVert}

\newcommand{\Xspace}{\mathcal{X}}
\newcommand{\median}{\mu^P}

\providecommand{\eqdef}{:=}

\DeclareMathOperator{\dist}{dist}
\DeclareMathOperator*{\argmin}{argmin}
\DeclareMathOperator{\cost}{cost}
\DeclareMathOperator{\rcda}{RCDA}

\usepackage{xcolor}

\newtheorem{theorem}{Theorem}
\numberwithin{theorem}{section}
\newtheorem{lemma}[theorem]{Lemma}
\newtheorem{corollary}[theorem]{Corollary}
\newtheorem{definition}[theorem]{Definition}
\newtheorem{proposition}[theorem]{Proposition}
\newtheorem{claim}[theorem]{Claim}

\theoremstyle{definition}
\newtheorem{remark}[theorem]{Remark}

\title{
Optimal Stable Coresets for Geometric Median via Uniform Sampling
}

\author{Amir Carmel%
    \thanks{Email: \texttt{amir6423@gmail.com}
    }
    \qquad
    Robert Krauthgamer%
        \thanks{The Harry Weinrebe Professorial Chair of Computer Science.
            Work partially supported by the Israel Science Foundation grant \#1336/23. 
            Email: \texttt{robert.krauthgamer@weizmann.ac.il}
        }
    \qquad
    Nir Petruschka%
    \thanks{Email: \texttt{nir.petruschka@weizmann.ac.il}
    }
    \\  Weizmann Institute of Science
}
\date{}
\begin{document}

\maketitle

\begin{abstract}
The geometric median problem asks to find a point in $\Rd$ that minimizes the sum of Euclidean distances to an input set.
It is a classical problem in computational geometry and appears as a subroutine in numerous optimization tasks, 
many of which require the solution to satisfy additional structural constraints.
A common approach to reduce the input size is to construct a coreset,
which is a small weighted subset that faithfully represents the input for a specific optimization problem.
Strong coresets preserve the cost of every candidate solution but require linear time to construct; 
weak coresets admit sublinear construction, in fact by uniform sampling, 
but only preserve near-optimal solutions, which is insufficient when the solution is constrained. 
To address this, we focus instead on the recently introduced intermediate notion of a \emph{stable coreset}, which simultaneously handles all constrained variants.
Currently, there is a large gap between the known sample sizes for stable and weak coresets.

Our main result is that a uniform sample of size $O(\epsilon^{-2} \log \tfrac{1}{\epsilon})$ is a stable $(\epsilon, O(\epsilon))$-coreset for the geometric median, with high constant probability, 
and this bound is tight up to the logarithmic factor.
Our analysis adapts recent machinery of Carmel and Krauthgamer (ICLR 2026) 
for constructing stable coresets, which incurs an $O(\log d)$ factor.
We show an iterative argument that progressively reduces the sample size, 
and eliminates this dependence on the dimension $d$. 
At a high level, this approach resembles the technique of iterative size reduction, 
which is applicable for strong coresets but not for weak coresets.
\end{abstract}

\section{Introduction}

Clustering is a basic tool for understanding large collections of data, with applications throughout machine learning, statistics, and optimization. 
The core idea is to partition the data points into sets (called clusters) that share common characteristics,
and a classical formulation is center-based clustering, 
where each cluster is represented by a single point. 
One of the most fundamental problems in this family is the \emph{geometric median}, 
which asks for a point that minimizes the sum of distances to all the input points. 
Despite its apparent simplicity, the geometric median exhibits strong robustness properties 
and has broad applications, from robust statistics to facility location and machine learning. 
To define this problem formally, we view it as a special case of the \emph{$k$-median problem}, 
where the input is a point set $P$ in a metric space $(\Xspace, \dist)$, 
and the goal is to find a set $C \subseteq \Xspace$ of $k$ centers 
that minimize the objective
$$\cost(C, P) \eqdef \sum_{p \in P} \min_{c \in C} \dist(p, c).$$
The \emph{geometric median problem} (also called the \emph{Fermat-Weber problem})
is the case $k=1$ in the widely studied setting of Euclidean spaces,
i.e., $\Xspace = \Rd$ and $\dist(x,y) = \norm{x-y}_2$. 
Our main result addresses the geometric median problem, however the broader context is often useful, 
e.g., our proof connects this problem to $1$-median in $\ell_1$ spaces. 

As datasets grow in size, data reduction techniques have become essential, as they enable algorithms to operate on small summaries rather than on the entire dataset. 
A coreset formalizes this idea by providing a small instance of the same problem, 
where solving the problem on the coreset immediately provides similar guarantees for the original dataset.
Coresets usually offer a tunable tradeoff between size and quality, the latter measured as an approximation factor. 
A \emph{strong $\epsilon$-coreset} for $k$-median is a weighted subset $Q \subseteq P$ 
such that for every set of $k$ centers $C \subset \Rd$, 
the cost is approximately preserved, namely, 
\begin{equation}
  \cost(C, Q) \in (1 \pm \epsilon) \cdot \cost(C, P). 
\end{equation}
Strong coresets for the classical $k$-median problem in Euclidean space, alongside numerous variants, have been studied extensively over the past two decades, 
see \cref{sec:related} for more details.

Many modern applications involve massive datasets,
for which even linear-time processing might be prohibitive.
Yet, constructing a strong coreset inherently requires reading each data point,
which takes at least linear time $\Omega(nd)$, regardless of the coreset's size.%
\footnote{Consider an instance $P$ with a single extremely distant point. 
A strong coreset must include this point, which requires scanning the entire instance.}
Therefore, to open the door for sublinear-time algorithms,
one natural relaxation is the notion of a \emph{weak coreset}.
A weighted subset $Q \subseteq P$ is called a weak $(\epsilon, \eta)$-coreset 
if every near-optimal solution for $Q$ is also a near-optimal solution for $P$,
namely, for every set of $k$ centers $C \subset \Rd$, 
\begin{equation}
  \text{if } \cost(C, Q) \leq (1+\epsilon) \cdot \min_{|C'| = k} \cost(C', Q) ,
  \text{ then }
  \cost(C, P) \leq (1+\eta) \cdot \min_{|C'| = k} \cost(C', P). 
\end{equation}
It is worth noting that most constructions achieve $\eta=O(\epsilon)$, 
but covering the entire literature requires two separate parameters.
Weak coresets can be constructed via uniform sampling, a method widely adopted in practice, often without theoretical guarantees, and naturally suited to streaming and distributed settings.
See \cref{tab:comparison} for known bounds and references.

Despite its appeal, the weak coreset notion has a fundamental limitation: 
A strong coreset has guarantees for \emph{every} candidate solution,
but a weak coreset provides guarantees only for candidates that are near-optimal (on the coreset).
This might be a severe limitation when the solutions of interest are constrained to a subset $\Xspace'\subset \Xspace$
(because an optimal solution in $\Xspace'$ need not be optimal also in $\Xspace$),
and such constraints indeed arise from structural requirements or downstream tasks, as we discuss further below. 

The notion of a \emph{stable coreset}, recently introduced by~\citet{carmel2025stablecoresetsunleashingpower}, 
is intermediate between the strong and weak coreset notions and addresses this issue,
by preserving the relative quality of \emph{all pairs} of candidate solutions.
Formally, a weighted subset $Q \subseteq P$ is called a stable $(\epsilon, \eta)$-coreset 
if for every pair of candidates $C, C' \subset \Rd$, 
\begin{equation}
  \text{if } \cost(C, Q) \leq (1+\epsilon) \cdot \cost(C', Q) , 
  \text{ then }
  \cost(C, P) \leq (1+\eta) \cdot \cost(C', P).
\end{equation}
Crucially, since this guarantee holds simultaneously for all pairs of candidates, 
the same stable coreset can be reused across multiple constraints and downstream tasks, 
and its construction does not require a priori knowledge of the constraints.

The geometric median problem is used as a subroutine in numerous optimization tasks, 
and many of them require the solution to satisfy additional constraints,
which essentially restrict it to a predetermined subset $\Xspace'\subset \Xspace = \Rd$.
For instance, the solution may be required to lie in a low-dimensional linear subspace, 
or more generally on an algebraic variety, 
or to have a sparse coordinate structure (e.g., $\|c\|_0\leq \alpha$).

We illustrate this with three concrete examples.
First,  in kernel-based methods, the input points are all images of a feature map $\varphi$  
and also the solution must be an image of $\varphi$~\citep{DBLP:conf/icml/MinskerSLD14, DBLP:journals/pami/NienkotterJ23}; 
in this case, $\Xspace'$ is simply the image set of $\varphi$.
Second,  the \emph{discrete} geometric median problem requires the solution to be one of the input points in $P$~\citep{DBLP:conf/aistats/NewlingF17, DBLP:journals/algorithmica/Har-PeledJR21, DBLP:conf/cccg/DaescuT22}; in this case $\Xspace'=P$ is discrete.
Third, the \emph{ultrametric $1$-median} problem, in which the input distances form an ultrametric and the solution is a median terminal node, admits a $(1+\epsilon)$-approximation in $\widetilde{O}(\epsilon^{-3})$ time~\citep{CHANG202065}. Since finite ultrametrics embed isometrically into $\ell_2$~\citep{SHKARIN200413}, this is a special case of constrained $\ell_2$ $1$-median.

In all these cases where we seek an optimal or near-optimal solution in $\Xspace'$, a weak coreset provides no guarantees,
while a stable coreset is applicable, in fact simultaneously for all possible restrictions $\Xspace'$.

\renewcommand{\arraystretch}{1.2}
\begin{table}[t]
\centering
\caption{Sublinear-time constructions of coresets for the geometric median problem. 
In all of them, the sample complexity equals the (listed) coreset size;
in fact, the coreset is simply the sample, except when marked by $\dagger$, in which the sample is post-processed.
Strong coresets are not listed as they require at least linear time. 
}
\begin{tabular}{|l|cl|}
\hline
Coreset type & Coreset size & Reference \\
\hline
\multirow{2}{*}{weak $(0,\epsilon)$} 
  & $\widetilde{O}(\epsilon^{-2} d)$ & \cite{DBLP:journals/talg/AckermannBS10,DBLP:journals/ki/MunteanuS18} \\
\cline{2-3}
  & $\widetilde{O}(\epsilon^{-4})$ & \cite{danos21} \\
\hline
\multirow{3}{*}{weak $(\epsilon,O(\epsilon))$} 
  & $\widetilde{O}(\epsilon^{-4})\ ^\dagger$ & \cite{DBLP:conf/nips/Cohen-AddadSS21} \\
\cline{2-3}
  & $\widetilde{O}(\epsilon^{-3})$ & \cite{DBLP:conf/icml/HuangJL23} \\
\cline{2-3}
  & $\widetilde{O}(\epsilon^{-2})$ & \cite{DBLP:conf/icml/Woodruff024a} \\
\hline
\multirow{2}{*}{stable $(\epsilon,O(\epsilon))$} 
  & $\widetilde{O}(\epsilon^{-2}\log d)$ & \cite{carmel2025stablecoresetsunleashingpower} \\
\cline{2-3}
  & $\widetilde{O}(\epsilon^{-2})$ & \cref{thm:main}\\
\hline
\end{tabular}
\label{tab:comparison}
\end{table}

\subsection{Our Results}
\label{sec:results}

A long line of work has progressively refined the size of weak coresets (see~\cref{tab:comparison}), 
however these results and techniques do not extend (at least not directly) to stable coresets,
for which the known bound falls behind and even depends on the dimension $d$.
Our main result is a stable coreset for geometric median 
that is constructed via uniform sampling (with no post-processing)
and has near-optimal size and query complexity. 
In our result, a uniform sample is one drawn without replacement.

\begin{restatable}{theorem}{mainthm}\label{thm:main}
Let $\epsilon \in (0,1/4)$ and let $P\subset\Rd$ be finite. 
Then a uniform sample of size $O(\epsilon^{-2} \log \frac{1}{\epsilon})$ from $P$ 
is a stable $(\epsilon, 17\epsilon)$-coreset for $1$-median in $\ell_2^d$ 
with probability at least $4/5$.
\end{restatable}

This result improves over the stable coreset of \citet{carmel2025stablecoresetsunleashingpower}, 
which is also via uniform sampling and has size $\widetilde{O}(\epsilon^{-2}\log d)$. 
Our bound completely eliminates the dependence on the dimension.
Technically, their main result is for $\ell_1$ metrics, 
and the Euclidean case follows by Dvoretzky’s theorem~\citep{Gordon1988,schechtman2006remark}.
Previously, a similar coreset size $\widetilde{O}(\epsilon^{-2})$, also via uniform sampling, 
was known only for weak coreset~\citet{DBLP:conf/icml/Woodruff024a}, which is strictly less powerful. 
Technically, it is obtained as a corollary of a result for multiple $\ell_p$ regression, 
whose proof leverages $\ell_p$ Lewis weight sampling. 
In contrast, our proof refines the framework for stable coresets recently developed by~\citet{carmel2025stablecoresetsunleashingpower},
and employs an iterative argument to progressively reduce the sample size, see \cref{sec:overview}.

The dependence on $\epsilon$ in our bound is near-tight, 
as an $\Omega(\epsilon^{-2})$ lower bound follows from more general results 
of~\citet{DBLP:conf/colt/ChenD21, DBLP:conf/approx/ParulekarPP21}.
For completeness, we provide a short self-contained proof in Appendix~\ref{app:lowerbound} .

\begin{remark}
We consider only the case $k=1$ in order to obtain clean theoretical guarantees. 
It is well-known that sublinear algorithms face a major challenge when $k>1$,
because small clusters can be missed by sublinear methods.
Prior work by~\citet{DBLP:conf/icml/HuangJL23} (see also references therein)
has formalized this issue through a \emph{balancedness parameter} $\beta \in (0,1]$, which measures the smallest cluster size in an optimal solution, modifying the coreset definition to consider only $\beta$-balanced clusterings of the coreset. 
\end{remark}

\subsection{Technical Overview}
\label{sec:overview}

We now overview the proof of our main result, \cref{thm:main}. 
Our starting point is a result of \citet[Corollary D.2]{carmel2025stablecoresetsunleashingpower}, 
which establishes that a uniform sample of size $O(\epsilon^{-2}\log (d/\epsilon))$ is a stable $(\epsilon, O(\epsilon))$-coreset for $1$-median in $\ell_2^d$ with high constant probability.  
Their proof exploits the fact that $\ell_2^d$ embeds into $\ell_1^{O(d)}$, by Dvoretzky's theorem, 
and applies their analysis for the 1-median problem in $\ell_1$. 
We make a simple observation, that a folklore terminal embedding (\cref{thm:isometric_embedding})
reduces the problem to the case where $d \leq n$,
and therefore a uniform sample $Q$ of size $O(\epsilon^{-2}\log (n/\epsilon))$ suffices for such a coreset, regardless of the ambient dimension $d$ of our dataset $P$.

Building on this simple observation, our high-level idea is to apply iterative size reduction, 
a technique introduced by~\citet{DBLP:conf/soda/BravermanJKW21} for strong coresets, 
hence we need to adapt it to stable coresets.
Suppose we have a uniform sample $Q \subset P$ of size $\widetilde{O}(\epsilon^{-2}\log n)$, 
and then we pick from it a uniform sample $R \subseteq Q$ 
of size $\widetilde{O}(\epsilon^{-2}\log |Q|) = \widetilde{O}(\epsilon^{-2}\log\log n)$. 
By our preceding discussion, with constant high probability, 
$Q$ is a stable coreset of $P$, and also $R$ is a stable coreset of $Q$.
Now if coresets can be composed (i.e., a coreset of a coreset is a coreset), 
which is indeed the case for strong coresets,
then $R$ is also a coreset of our initial dataset $P$. 
The crux is to apply this argument iteratively until the size of the sample no longer decreases,
and by solving $|Q| \leq \epsilon^{-2}\log (|Q|/\epsilon)$
we can easily see that the final sample size is $\widetilde{O}(\epsilon^{-2})$. 
Importantly, while this analysis goes through cumbersome iterations, the algorithm is simple:
A uniform sample $R$ from $Q$, which itself is a uniform sample from $P$,
amounts to a uniform sample $R$ taken directly from $P$.
Therefore, the final coreset is just a uniform sample of size $\widetilde{O}(\epsilon^{-2})$ from the input $P$.

Unfortunately, the above strategy has considerable gaps. 
One issue is that each iteration might fail with a constant probability,
which is too large because the number of iterations is typically $O(\log^* n)$. 
One can decrease this probability by tuning parameters differently
(standard amplification by majority over repetitions would deviate from uniform sampling),
but the final coreset size will depend on $n$, missing our intended goal. 
Another issue is that the errors accumulate over the iterations.
Prior work \citep{DBLP:conf/soda/BravermanJKW21, DBLP:conf/nips/Cohen-AddadSS21, DBLP:conf/stoc/Huang0024}
addressed this by setting a different approximation parameter $\epsilon_i$ for each iteration $i$.
Typically, one starts with a small value, say $\epsilon_1=1/\log n$, 
and increases it exponentially fast until some final value $\epsilon_\text{final}$,
which determines the final coreset size.
For strong coresets, the accumulated error is 
$\prod_i (1+\epsilon_i) \le 1 + \sum_i O(\epsilon_i) \le 1+O(\epsilon_\text{final})$,
however for stable coresets the errors compound completely differently.
Suppose that each iteration yields a stable $(\epsilon, 2\epsilon)$-coreset, 
then we must set $\epsilon_{i+1} = \epsilon_{i}/2 = \epsilon_{1}/2^i$, 
which results overall in a stable $(\epsilon_{\mathrm{final}}, \epsilon_1)$-coreset.
As the number of iterations grows with $n$, albeit slowly like $O(\log^* n)$, 
this guarantee is rather poor and certainly misses our intended goal.

Since iterative size reduction cannot be applied in a black-box manner,
we need to implement the iterations in a more specialized manner.
We identify that (a slight adaptation of) 
a notion called \emph{Relative Cost-Difference Approximation} (RCDA), 
that was introduced by~\cite{carmel2025stablecoresetsunleashingpower}, 
avoids the core issue in the scheme above. 
Informally, the RCDA notion asserts that the difference in costs between any two candidate centers 
is approximately the same when measured on the sample versus on the full dataset.
We compose this RCDA notion over the iterations, which requires some technical adaptations.
It also introduces some error terms and failure probabilities, 
and tracking these carefully over the iterations is somewhat intricate. 
Our main technical lemma (\cref{lemma:iterative_rcda}) shows that the final sample is an RCDA for $P$,
which, in turn, implies that it forms a stable coreset.

\subsection{Related Work}
\label{sec:related}

Strong coresets, which are often simply referred to as coresets, have been extensively studied for the $k$-median and $k$-means problems in Euclidean metrics. 
After a long line of works, which includes~\citep{DBLP:conf/stoc/Har-PeledM04, DBLP:journals/siamcomp/Chen09,feldman2011unified, DBLP:conf/focs/SohlerW18,DBLP:conf/stoc/HuangV20, DBLP:conf/soda/BravermanJKW21, DBLP:conf/stoc/Cohen-AddadSS21}, 
the current state-of-the-art constructions for $k$-median and $k$-means 
match the known lower bound of $\Omega(k\epsilon^{-2})$, up to polylogarithmic factors~\citep{DBLP:conf/stoc/Cohen-AddadLSS22,DBLP:conf/stoc/Huang0024}.
For the $\ell_1$ metric, a coreset of size $\operatorname{poly}(k/\epsilon)$ 
follows from~\citep{DBLP:journals/ml/JiangKLZ24} by going through $\ell_2$-squared.
The case of $k=1$ has received special attention. In low-dimensional Euclidean spaces
\citet{DBLP:conf/icml/HuangHH023} obtain a strong coreset of size $\widetilde{O}(\sqrt{d}/\epsilon)$, 
which was subsequently improved by~\citet{DBLP:conf/icml/AfshaniS24} to $\widetilde{O}(\epsilon^{-d/(d+1)})$ for any constant dimension $d$.

Many coreset variants of the $k$-median and $k$-means problems have been studied as well. 
Closest to our work is the line of weak coresets for $1$-median in Euclidean spaces, 
where the construction can be performed in sublinear time; 
see \cref{tab:comparison} for a detailed comparison of previous results.
Additionally, coresets have been studied for other variants of the geometric median problem 
and the more general $k$-median problem in Euclidean metrics, 
including capacitated $k$-median~\citep{DBLP:conf/focs/BravermanCJKST022,DBLP:conf/soda/Huang0L025}, 
fair clustering~\citep{DBLP:conf/waoa/0001SS19, DBLP:conf/nips/HuangJV19,pmlr-v151-chhaya22a}, 
robust clustering~\citep{HuangJL023,DBLP:conf/icalp/JiangL25, fang2025coresetrobustgeometricmedian}, 
fuzzy clustering~\citep{DBLP:conf/isaac/BlomerBB18}, 
ordered weighted clustering~\citep{DBLP:conf/icml/BravermanJKW19}, 
and time-series clustering~\citep{DBLP:conf/nips/HuangSV21}.

Another closely related problem is the $1$-center, also known as the minimum enclosing ball problem. Strong coresets for this problem in Euclidean metrics require exponential dependence on the dimension~\citep{AHV05}. Weak coresets, in contrast, were studied by~\citet{BHI02,BC08}, who show a tight size bound of $\lceil 1/\epsilon \rceil$.
For the $\ell_1$ metric, \citet{DBLP:conf/innovations/CarmelGJK25} studied both strong and weak coresets, and showed that the coreset size must depend on the dimension.

Lastly, a separate line of work studies sublinear approximation algorithms for the geometric median (or the related $1$-means objective), rather than constructing a coreset~\citep{DBLP:conf/stoc/CohenLMPS16,bertolottisimple}.

\section{Preliminaries}\label{section:preliminaries}

In the 1-median problem on the metric space $\ell_q^d =(\Rd, \norm{\cdot}_q)$, we are given a set $P \subset \Rd$ and the goal is to find a point $x\in\Rd$ that minimizes the sum of distances to $P$, that is, to minimize the cost function
\[
\cost_q(x,P) = \frac{1}{|P|}\sum_{p \in P} \norm{x-p}_q.
\]

We denote by $\mu^P$ an optimal solution to the median problem on $P$, $\median = \argmin_{x \in \Rd} \cost_q(x,P)$.

In this paper, we focus on the 1-median problem in the Euclidean metric, that is, $q=2$. Throughout our analysis, we leverage metric embeddings to translate results between $\ell_2$ and $\ell_1$.

The following theorem is a classical consequence of Gordon’s refinement of Dvoretzky’s theorem~\citep{Gordon1988,schechtman2006remark}.

\begin{theorem}[Dvoretzky’s theorem]\label{thm:gordon-dvoretzky}
For every $\gamma \in (0,1)$ and natural number $d$, there exists a linear map $f: \RR^d \to \RR^m$ with $m = O(\gamma^{-2} d)$ such that
$$\forall  x \in \RR^d, \qquad \norm{f(x)}_1 \in  (1 \pm \gamma) \norm{x}_2.$$
\end{theorem}

Given such embedding $f$, it immediately follows that the cost function is also preserved up to $1 \pm \gamma$. For a set $P \subset \Rd$, we denote by $f(P) = \{f(p) : p \in P\}$ the image of $P$ under $f$.

\begin{corollary}\label{obs:cost_embedding}
Let $f:\RR^d \to \RR^m$ be as in Theorem~\ref{thm:gordon-dvoretzky} with parameter $\gamma$ and let $P \subset \Rd$. Then
$$\forall x \in \Rd, \qquad  \cost_1 (f(x),f(P)) \in (1 \pm \gamma)\cost_2 (x,P).$$
\end{corollary}

We will need the following folklore isometric terminal embedding in $\ell_2$. 
Although variants of it appear in \citep{DBLP:conf/approx/ElkinFN15,DBLP:conf/stoc/HuangV20}, 
we include a proof for completeness.

\begin{theorem}[Folklore]\label{thm:isometric_embedding}
Let $P \subset \Rd$ of size $n$. Then there exists an isometric terminal embedding $f:\Rd \rightarrow \RR^n$, that is, for all $x \in \Rd$ and $p \in P$, $\norm{f(x)-f(p)}_2  =\norm{x-p}_2$.
\end{theorem}

\begin{proof}
Let $P=\{p_1,\dots,p_n\} \subset \Rd$ be a terminal set. By shifting the points, we can assume w.l.o.g. that $p_n$ lies at the origin. Consider the $(n-1)$-dimensional subspace $V$ spanned by the vectors $p_1,\dots,p_{n-1}$. As every $d$-dimensional subspace of $\ell_2$ is isometric to $\mathbb{R}^d$, we can also assume w.l.o.g. that $V=\mathbb{R}^{n-1}$. Let $\Pi : \Rd \to V$ be the orthogonal projection map onto $V$ and define the terminal embedding $f: \Rd \to V \times \mathbb{R} = \mathbb{R}^{n}$ by $f(z)=(\Pi(z), \|\Pi(z)-z\|_2)$. Using the Pythagorean theorem, we conclude that for every $p \in P \subset V$ and $x \in \Rd$
\[
    \|f(p)-f(x)\|_2^2= \|\Pi(p)-\Pi(x)\|_2^2+(\|\Pi(x)-x\|_2-0)^2 = \|p-\Pi(x)\|_2^2+\|\Pi(x)-x\|_2^2 =\|p-x\|_2^2.
\]
\end{proof}

\section{RCDA in \texorpdfstring{$\ell_2$}{l2} via Dimension Reduction}\label{section:rcda}

The RCDA property was formulated in~\citep{carmel2025stablecoresetsunleashingpower}. Here we introduce a slightly modified definition.

\begin{definition}[Relative Cost-Difference Approximation]\label{def:rcda}
We say that $Q \subseteq P$ is an \emph{$\epsilon$-RCDA$_q$} of $P$ if
\begin{equation}\label{eq:rcda}
  \forall x,y \in \Rd, 
  \qquad
  \Big| \big[\cost_q(x,P) - \cost_q(y,P)\big]
    - \big[\cost_q(x,Q) - \cost_q(y,Q)\big] \Big|
  \leq \epsilon\cdot (\cost_q(x,P) +\cost_q(y,P)).
\end{equation}
We will denote the left-hand side by $\rcda_q(x,y,Q,P)$ for brevity.
\end{definition}

Our definition differs from that in~\citep{carmel2025stablecoresetsunleashingpower} in that \eqref{eq:rcda} quantifies over all pairs $x,y \in \Rd$, whereas the original property was defined with respect to $x$ and the median $\median$. Specifically, the original definition required
\begin{equation}\label{eq:old_rcda}
\forall x\in\Rd, \qquad \rcda_q(x,\median,Q,P) \leq \epsilon \cost_q(x,P).
\end{equation}

However, the two definitions are essentially equivalent. Indeed, if for every $x \in \Rd$, \\ $\rcda_q (x,\mu^P,Q,P) \leq \epsilon \cdot \cost_q(x,P)$, then for every $x,y \in \Rd$,
\[
\rcda_q (x,y,Q,P) \leq \rcda_q (x,\mu^P,Q,P) + \rcda_q (\mu^P,y,Q,P) \leq \epsilon \cdot (\cost_q(x,P) + \cost_q(y,P)),
\]
where the first inequality follows from the triangle inequality. Conversely, if for every $x,y \in \Rd$, we have $\rcda_q (x,y,Q,P) \leq \epsilon \cdot (\cost_q(x,P) + \cost_q(y,P))$, then for every $x \in X$,
\[
\rcda_q (x,\mu^P,Q,P) \leq \epsilon \cdot (\cost_q(x,P) + \cost_q(\mu^P,P)) \leq 2\epsilon \cdot \cost_q(x,P).
\]

The main technical result from \citep{carmel2025stablecoresetsunleashingpower} establishes the RCDA property for a uniform sample in $\ell_1^d$.

\begin{lemma}[Lemma~4.6 in \citep{carmel2025stablecoresetsunleashingpower}]\label{lem:rcda_in_ell1}
Let $\epsilon \in (0,1)$ and let $P \subset \Rd$ of size $n$.
Then a uniform sample $Q \subseteq P$ of size $O(\epsilon^{-2} \log \frac{d}{\delta})$ is an $\epsilon$-RCDA$_1$ of $P$ in $\ell_1^d$ with probability at least $1-\delta$.
\end{lemma}
\begin{remark}
While \cite{carmel2025stablecoresetsunleashingpower} does not explicitly state whether the uniform sample is taken with or without replacement, their result holds in both cases.
The only step in their proof that depends on the sampling scheme is a theorem from \citet{DBLP:journals/jcss/LiLS01} on $\epsilon$-approximations, which applies equally to sampling with and without replacement~\citep{CM22}.
\end{remark}

A direct consequence of the RCDA property, is a uniform bound over all costs.

\begin{proposition}\label{prop:cost_inflation}
Let $\epsilon \in (0,1/2)$ and $\alpha \geq 1$. If $Q$ is an $\epsilon$-RCDA$_q$ of $P$ and $\cost_q(\median,Q) \leq \alpha \cost_q(\median,P)$, then,
$$\forall x \in \Rd, \qquad \cost_q(x,Q) \leq 2\alpha \cost_q(x,P).$$
\end{proposition}

\begin{proof}
By Definition~\ref{def:rcda}, we have $\rcda_q(x,\median,Q,P) \leq \epsilon\cdot (\cost_q(x,P) +\cost_q(\median,P))$.
Therefore:
\begin{align*}
    \cost_q(x,Q) &\leq \cost_q(x,P) + \cost_q(\median,Q) - \cost_q(\median,P) + \epsilon(\cost_q(x,P)+\cost_q(\median,P))\\
    &\leq (1+\epsilon)\cost_q(x,P) + (\alpha-1+\epsilon)\cost_q(\median,P)\\
    &\leq (1+\epsilon)\cost_q(x,P) + (\alpha-1+\epsilon)\cost_q(x,P) \tag{$\median$ is optimal}\\
    &= (\alpha + 2\epsilon)\cost_q(x,P) \leq 2\alpha\cost_q(x,P). \tag{$2\epsilon \leq 1 \leq \alpha$}
\end{align*}
\end{proof}

We now establish that uniform sampling yields RCDA in $\ell_2^d$ by reducing to the $\ell_1$ case.

\begin{lemma}\label{lemma:rcda_in_ell2}
Let $\epsilon \in (0,1)$, $\alpha \geq 1$, and let $P \subset \Rd$ of size $n$.
Then, with probability at least $1-1/\alpha-\delta$, a uniform sample $Q \subseteq P$ of size $O(\epsilon^{-2} \log \frac{\alpha \cdot n}{\epsilon \cdot \delta})$ satisfies:
\begin{enumerate}
    \item $Q$ is an $\epsilon$-RCDA$_2$ of $P$ in $\ell_2^d$, and,
    \item For every $x\in \Rd$, $\cost_2(x,Q) \leq 3\alpha \cost_2(x,P)$.
\end{enumerate}
\end{lemma}

\begin{proof}
By Theorem~\ref{thm:isometric_embedding}, we may assume without loss of generality that $d=n$.

Let $\gamma = \frac{\epsilon}{10\alpha}$. By Theorem~\ref{thm:gordon-dvoretzky}, there exists $m=O(\alpha^2 \epsilon^{-2} n)$ and a linear map $f:\RR^d \rightarrow \RR^m$ such that
for all $x \in \RR^d$, $\norm{f(x)}_1 \in (1\pm \gamma) \norm{x}_2.$

By Lemma~\ref{lem:rcda_in_ell1}, a uniform sample $Q \subseteq P$ of size 
$O(\epsilon^{-2} \log \frac{m}{\delta}) = O(\epsilon^{-2} \log \frac{\alpha n}{\epsilon \delta})$
is a $\frac{\epsilon}{2}$-RCDA$_1$ of $f(P)$ in $\ell_1^m$ with probability at least $1-\delta$. That is,
\begin{equation}\label{eq:rcda_ell1}
\forall x,y \in \RR^d, \qquad \rcda_1(f(x),f(y),f(Q),f(P)) \leq \frac{\epsilon}{2} (\cost_1(f(x),f(P)) +\cost_1(f(y),f(P))).
\end{equation}

Let $\median$ be the optimal median of $P$ in $\ell_2^d$.
Since $\mathbb{E}_Q[\cost_1(f(\median),f(Q))] = \cost_1(f(\median),f(P))$, by Markov's inequality, with probability at least $1-1/\alpha$,
\begin{equation}\label{eq:markov}
    \cost_1 (f(\median),f(Q)) \leq \alpha \cost_1(f(\median),f(P)). 
\end{equation}

In this event, by Proposition~\ref{prop:cost_inflation} applied in $\ell_1^m$ (noting $\epsilon/2 < 1/2$), 
\begin{equation}\label{eq:cost_in_ell1}
\forall x\in \Rd, \qquad \cost_1(f(x),f(Q)) \leq 2\alpha \cost_1 (f(x),f(P)).    
\end{equation}

Using Corollary~\ref{obs:cost_embedding}, for all $x \in \RR^d$,
\begin{align*}
    \cost_2(x,Q) &\leq (1+\gamma) \cost_1(f(x),f(Q)) \leq (1+\gamma) \cdot 2\alpha \cost_1 (f(x),f(P)) \\
    &\leq \frac{2\alpha(1+\gamma)}{1-\gamma} \cost_2(x,P) \leq 3\alpha \cost_2 (x,P),
\end{align*}
where the last inequality uses $\gamma = \epsilon/(10\alpha) \leq 1/10$, establishing property (2).

To show property (1), let $x,y \in \RR^d$. 
Again using Corollary~\ref{obs:cost_embedding} applied to each of the four cost terms and the triangle inequality:
\begin{align*}
  \rcda_2(x,y,Q,P) 
  &\leq \rcda_1(f(x),f(y),f(Q),f(P)) + \gamma(\cost_1 (f(x),f(P)) + \cost_1 (f(y),f(P)) \\
  &\qquad +\cost_1 (f(x),f(Q)) +\cost_1 (f(y),f(Q)))
\end{align*}

By the union bound  \eqref{eq:rcda_ell1} and \eqref{eq:cost_in_ell1} holds with probability at least $1-1/\alpha-\delta$. We get:

\begin{align*}
  \rcda_2(x,y,Q,P) 
  &\leq \frac{\epsilon}{2} (\cost_1(f(x),f(P)) +\cost_1(f(y),f(P))) \\
  &\quad + \gamma(2\alpha+1) (\cost_1(f(x),f(P)) + \cost_1(f(y),f(P)))\\
  &\leq \frac{1}{1-\gamma} \left(\frac{\epsilon}{2} + \gamma(2\alpha+1)\right) (\cost_2(x,P) +\cost_2(y,P))\\
  &\leq \frac{10}{9}\left(\frac{\epsilon}{2} + \frac{\epsilon(2\alpha+1)}{10\alpha}\right) (\cost_2(x,P) +\cost_2(y,P))\\
  &\leq \frac{10}{9}\left(\frac{\epsilon}{2} + \frac{3\epsilon}{10}\right) (\cost_2(x,P) +\cost_2(y,P)) \leq \epsilon (\cost_2(x,P) +\cost_2(y,P)).
\end{align*}
\end{proof}

\section{Upper Bound (Proof of~\cref{thm:main})}
\label{section:mainthm}

To apply the iterative size reduction technique, we first establish the following composition property for RCDA, which allows errors to be tracked across iterations.

\begin{proposition}\label{prop:rcda_composition}
Let $Q'' \subseteq Q' \subseteq Q$ with $\epsilon, \epsilon' \in (0,1)$ and $\beta \geq 1$. 
Suppose $Q'$ is an $\epsilon$-RCDA$_q$ of $Q$ and $Q''$ is an $\epsilon'$-RCDA$_q$ of $Q'$. If for every $x \in \Rd$ we have $\cost_q(x,Q') \leq \beta \cost_q(x,Q)$, then, $Q''$ is an $(\epsilon + \epsilon' \beta)$-RCDA$_q$ of $Q$.
\end{proposition}

\begin{proof}
For $x,y \in \Rd$, using the triangle inequality, 
\begin{align*}
\rcda_q(x,y,Q,Q'') &
\leq \rcda_q(x,y,Q,Q') + \rcda_q(x,y,Q',Q'')\\
&\leq \epsilon (\cost_q (x,Q) + \cost_q (y,Q)) + \epsilon'(\cost_q (x,Q') + \cost_q (y,Q'))\\
&\leq (\epsilon + \epsilon' \beta)(\cost_q (x,Q) + \cost_q (y,Q)).
\end{align*}
\end{proof}

We are now ready to prove our main technical lemma. The proof involves balancing the errors with the guaranteed size reduction.
To streamline the presentation, we use explicit constants and while we attempt to avoid huge constants,
they are probably far from optimal. 

\begin{lemma}\label{lemma:iterative_rcda}
Let $\epsilon \in (0,\frac{1}{4})$ and let $P \subset \Rd$ of size $n$.
Then, with probability at least $1-\frac{1}{30}$, a uniform sample $Q \subseteq P$ of size $O(\epsilon^{-2} \log \frac{1}{\epsilon  })$ is an $\epsilon$-RCDA$_2$ of $P$ in $\ell_2^d$.
\end{lemma}

\begin{proof}
By Theorem~\ref{thm:isometric_embedding}, we may assume $d=n$. 
Denote the $i$-th iterated logarithm (with base $2$) by $\log^{(i)}$, 
i.e., $\log^{(i+1)} n = \log(\log^{(i)} n)$ where $\log^{(0)} n = n$.
Let $t \geq 1$ be the smallest integer such that $\log^{(t-1)} n \leq 2048\epsilon^{-2}\log\epsilon^{-1}$;
then it also satisfies $\log^{(t-1)} n > \log (2048\epsilon^{-2}\log\epsilon^{-1}) \geq 16$.
For $i=1,\ldots,t$, define
\[
\epsilon_i = \frac{\epsilon}{1440(\log^{(i)}n)^{\frac{1}{2}}}, \qquad
\alpha_i = 180(\log^{(i)}n)^{\frac{1}{4}}, \qquad
\delta_i = 1/\alpha_i.
\]

We construct $Q_t \subseteq Q_{t-1} \subseteq \cdots \subseteq Q_1 \subseteq Q_0=P$, where each $Q_{i+1}$ is obtained from $Q_i$ by uniform sampling of size
\[
O\!\left(\epsilon_{i+1}^{-2}\log\!\left(\frac{\alpha_{i+1}n_i}{\epsilon_{i+1}\delta_{i+1}}\right)\right)
\]
with $n_i = |Q_i|$, as guaranteed by Lemma~\ref{lemma:rcda_in_ell2}. We next prove the following claim.

\begin{claim}
For every $i\in\{1,\ldots,t\}$, with probability at least
$1-3\sum_{j=1}^{i}\delta_j$,
\[
Q_i \text{ is an }
\epsilon\sum_{j=1}^i \frac{1}{4(\log^{(j)}n)^{\frac{1}{4}}}\text{-RCDA}_2
\text{ of } P.
\]    
\end{claim}
\begin{proof}
The proof proceeds by induction on $i\in\{1,\ldots,t\}$.

\emph{Base case.} By Lemma~\ref{lemma:rcda_in_ell2} with parameters $\epsilon_1, \alpha_1, \delta_1$, it follows that $Q_1$ is an $\epsilon_1$-RCDA$_2$ of $P$ with probability at least $1 -2 \delta_1$. 
The base case holds as  $\epsilon_1 = \epsilon /{1440(\log^{(1)}n)^{\frac{1}{2}}} \leq \epsilon /{4(\log^{(1)}n)^{\frac{1}{4}}}$.

\emph{Inductive step.} Since $Q_i$ is a uniform sample of $P$, by Markov's inequality, with probability at least $1-1/\alpha_{i+1}$, $$\cost_2(\median,Q_i) \leq \alpha_{i+1} \cost_2 (\median,P).$$ 

By Proposition~\ref{prop:cost_inflation}, for every $x\in \Rd$, $\cost_2(x,Q_i) \leq 2\alpha_{i+1} \cost_2(x,P)$. Thus, applying Lemma~\ref{lemma:rcda_in_ell2} to sampling $Q_{i+1}$ from $Q_i$ gives that, with probability at least $1-2\delta_{i+1}$, $Q_{i+1}$ is an $\epsilon_{i+1}$-RCDA$_2$ of $Q_i$. As by the induction hypothesis, $Q_i$ is an $\epsilon\sum_{j=1}^i \frac{1}{4(\log^{(j)}n)^{\frac{1}{4}}}\text{-RCDA}_2
\text{ of } P$ with probability at least $1-3\sum_{j=1}^{i}\delta_j$, both events hold with probability at least $1-3\sum_{j=1}^{i+1}\delta_j$ by a union bound. Assuming they both hold, we can now apply Proposition~\ref{prop:rcda_composition} and get, 
$$Q_{i+1} \text{ is an } \left(\epsilon\sum_{j=1}^i\frac{1}{4(\log^{(j)}n)^{\frac{1}{4}}} + \epsilon_{i+1} \cdot 2\alpha_{i+1}\right)\text{-RCDA}_2 \text{ of } P.$$ 
Since $\epsilon_{i+1} \cdot 2\alpha_{i+1} = \frac{\epsilon}{1440(\log^{(i+1)}n)^{\frac{1}{2}}}\cdot 360(\log^{(i+1)})^{\frac{1}{4}} = \frac{\epsilon}{4(\log^{(i+1)}n)^{\frac{1}{4}}}$, we have, $$\epsilon\sum_{j=1}^i\frac{1}{4(\log^{(j)}n)^{\frac{1}{4}}} + \epsilon_{i+1} \cdot 2\alpha_{i+1} \leq \epsilon\sum_{j=1}^{i+1}\frac{1}{4(\log^{(j)}n)^{\frac{1}{4}}},$$
concluding the induction step, and proving the claim.
\end{proof}

We proceed with the proof of Lemma~\ref{lemma:iterative_rcda}. 
By choice of $t$, for $1 \leq j \leq t-2$, we have $(\log^{(j+1)}n)^{\frac{1}{4}} \leq \frac{1}{8}(\log^{(j)}n)^{\frac{1}{4}}$ since $\log^{(j)}n \geq 2^{16}$. 
Thus,
$$\epsilon\sum_{j=1}^{t-1}\frac{1}{4(\log^{(j)}n)^{\frac{1}{4}}} \leq 
\frac{\epsilon}{4(\log^{(t-1)}n)^{\frac{1}{4}}}\sum_{j=0}^{t-2}8^{-j} \leq \frac{2\epsilon}{7(\log^{(t-1)}n)^{\frac{1}{4}}},$$
where the last inequality uses the standard bound over a sum of a geometric series.
Additionally, since $\log^{(t-1)} n \geq 16$, 
$$\epsilon\sum_{j=1}^{t}\frac{1}{4(\log^{(j)}n)^{\frac{1}{4}}} \leq \frac{2\epsilon}{7(\log^{(t-1)}n)^{\frac{1}{4}}} + \frac{\epsilon}{4(\log^{(t)}n)^{\frac{1}{4}}} \leq \frac{\epsilon}{2},$$
and similarly,
$$3\sum_{j=1}^t \delta_j \leq 
3\sum_{j=1}^{t-1} \delta_j + 3\delta_t
\leq \frac{24}{7}\delta_{t-1} + 3\delta_t \leq  \frac{24}{7\cdot180(\log^{(t-1)}n)^{\frac{1}{4}}} + \frac{3}{180(\log^{(t)}n)^{\frac{1}{4}}}
 \leq \frac{1}{105} + \frac{1}{60}
\leq \frac{3}{100},$$
where overall, we get that $Q_t$ is $\epsilon$-RCDA$_2$ of $P$ with probability at least $1-\frac{3}{100}$. 

\textbf{Sample size.}
It is easy to verify the invariant $
n_i \le C\,\epsilon^{-2} \log \epsilon^{-1}(\log^{(i)}n)^2  $
for all $i\in\{1,\ldots,t\}$,
for a sufficiently large constant $C$. Thus, at iteration $t$, since $\log^{(t-1)}n =O(\epsilon^{-2}\log\epsilon^{-1})$, we have that $
n_t = O(\epsilon^{-2}\log\epsilon^{-1}\cdot (\log (\epsilon^{-2}\log\epsilon^{-1}))^2)=O(\epsilon^{-2}\log^3 \frac{1}{\epsilon})$.

To further improve the sample size to $O(\epsilon^{-2}\log \frac{1}{\epsilon})$, we apply one additional iteration with constant parameters. 
Set $\epsilon_f = \frac{\epsilon}{4000}$, $\alpha_f = 1000$, and $\delta_f = 1/\alpha_f$, and uniformly sample $Q \subseteq Q_t$ of size 
$$O\left(\epsilon_f^{-2} \log \frac{\alpha_f n_t}{\epsilon_f \delta_f}\right) = O(\epsilon^{-2} \log \frac{1}{\epsilon}),$$
where note that, since sampling is uniform at every step, $Q$ is also a uniform sample of $P$.

Finally, by applying the same arguments as in the inductive step, we get that $Q$ is an 
$$\left(\frac{\epsilon}{2} + \epsilon_f \cdot 2\alpha_f\right)\text{-RCDA}_2 = \epsilon\text{-RCDA}_2$$
uniform sample of $P$ of size $O(\epsilon^{-2} \log \frac{1}{\epsilon})$ with probability at least $1 - \frac{3}{100} - (2/\alpha_f + \delta_f) \geq 1 - \frac{1}{30}$.
This concludes the proof of Lemma~\ref{lemma:iterative_rcda}. 
\end{proof}

We can now prove our main result, 
by plugging our notion of RCDA into the framework of \citep{carmel2025stablecoresetsunleashingpower}. 

\mainthm*

\begin{proof}
By Lemma~\ref{lemma:iterative_rcda}, a uniform sample $Q$ of size $O(\epsilon^{-2}\log \epsilon^{-1})$ is $\frac{\epsilon}{2}$-RCDA$_2$ of $P$ in $\ell_2^d$ with probability at least $29/30$. By Markov's inequality, with probability at least $5/6$, $\cost_2(\median,Q) \leq 6 \cost_2(\median,P)$. By the union bound, both events occur with probability at least $4/5$.
Now, let $x,y \in \Rd$ such that $\cost_2(x,Q) \leq (1+\epsilon)\cost_2(y,Q)$, 
then
\begin{align*}
\cost_2(x,P) \leq & \tfrac{\epsilon}{2} [\cost_2(x,P)+\cost_2(y,P)] + [\cost_2(x,Q)-\cost_2(y,Q)] +\cost_2(y,P) 
    \tag{$Q$ is $\frac{\epsilon}{2}$-RCDA$_2$ of $P$} \\
& \leq \tfrac{\epsilon}{2}\cost_2(x,P) + (1+\tfrac{\epsilon}{2})\cost_2(y,P) + \epsilon\cost_2(y,Q) 
    \tag{$\cost_2(x,Q) \leq (1+\epsilon)\cost_2(x,Q)$}\\
& \leq \tfrac{\epsilon}{2}\cost_2(x,P) + (1+\tfrac{\epsilon}{2}+12\epsilon)\cost_2(y,P) 
    \tag{by Proposition~\ref{prop:cost_inflation}, $\cost_2(y,Q) \leq 12\cost_2(y,P)$} \\
& \leq (1+\epsilon)(1+12.5\epsilon)\cost_2(y,P) \leq (1+17\epsilon)\cost_2(y,P).
\end{align*}
This establishes that $Q$ is a stable $(\epsilon,17\epsilon)$-coreset of $P$.
\end{proof}

\bibliographystyle{alphaurl}
\bibliography{bibliography}

\appendix
\section{Lower Bound}
\label{app:lowerbound}

\begin{restatable}{theorem}{lowerbound}\label{thm:lowerbound}
For every $\epsilon \in (0,1/4)$ 
and every metric space $(\mathcal{X}, \dist)$ containing at least two distinct points,
there exists a finite point set $P \subseteq \mathcal{X}$, 
such that every algorithm producing a weak $(0, \epsilon)$-coreset 
(and in particular weak $(\epsilon/2, \epsilon)$-coreset) 
for 1-median with probability at least $4/5$, 
has query complexity $\Omega(\epsilon^{-2})$.
\end{restatable}

\begin{proof}
Without loss of generality assume $\mathcal{X}$ contains two distinct points $x,y$ of distance $1$. For simplicity, assume $\epsilon/2$ divides $n$.

Consider two multisets $P_0$ and $P_1$. $P_0$ contains $n/2$ copies of $x$ and $n/2$ copies of $y$, consequently, $\cost(x,P_0) = \cost(y,P_0) = 1/2$. $P_1$ contains $n(1-\epsilon)/2$ copies of $x$ and $n(1+\epsilon)/2$ copies of $y$.
For $P_1$, $\cost(x,P_1) = (1+\epsilon)/2$ and $\cost(y, P_1) = (1-\epsilon)/2$, 
so $\cost(x,P_1) > (1+2\epsilon)\cost(y,P_1)$.

Let $Q$ be a multiset containing $s$ copies of $x$ and $m-s$ copies of $y$. Then $\cost(x,Q)/\cost(y,Q) = (m-s)/s$.
If $s \geq m/2$, then $\cost(x,Q) \leq \cost(y,Q)$.
Therefore, producing a weak coreset requires distinguishing whether $Q$ was sampled from $P_0$ or $P_1$. As a weak $(0,\epsilon)$-coreset of $P_1$ must have $y$ as an optimal solution. Under uniform sampling, $s \sim \text{Binomial}(m, 1/2)$ for $P_0$ and $s \sim \text{Binomial}(m, (1-\epsilon)/2)$ for $P_1$.
By an information-theoretic lower bound, distinguishing between two binomial distributions whose biases differ by $\epsilon$ with constant probability requires $m = \Omega(\epsilon^{-2})$ samples~\citep{chernoff1972sequential,DBLP:journals/jmlr/MannorT04,anthony2009neural}. 
\end{proof}

\end{document}